\documentclass[a4paper]{article}

\usepackage{amsmath}
\usepackage{amsthm}
\usepackage{amssymb}

\usepackage[left=3cm,right=3cm,
    top=3cm,bottom=3.5cm,bindingoffset=0cm]{geometry}

\usepackage[numbers, sort&compress]{natbib}

\usepackage{graphicx}

\theoremstyle{plain}
    \newtheorem{lemma}{Lemma}
    \newtheorem{theorem}{Theorem}
    \newtheorem{statement}{Proposition}

\theoremstyle{remark}

  \theoremstyle{definition}
    \newtheorem{definition}{Definition}
        \newtheorem{remark}{Remark}
                   \newtheorem{example}{Example}

\newcommand{\Ker}[1]{\mathrm{Ker} \, #1}
\newcommand{\rank}[1]{\mathrm{rank} \, #1}

\newcommand{\corank}[1]{\mathrm{corank} \, #1}

\newcommand{\diff}[1]{\mathrm{d}  #1}

\newcommand{\R}{\mathbb{R}}
\newcommand{\Complex}{\mathbb{C}}

\newcommand{\T}{\mathrm{T}}

\newcommand{\CP}{\overline{\mathbb{C}}}
\newcommand{\RP}{\overline{\mathbb{R}}}
\newcommand{\Hom}{\mathrm{H}}

\newcommand{\wave}{\widetilde}

\newcommand{\g}{\mathfrak{g}}

\newcommand{\so}{\mathfrak{so}}

\newcommand{\sL}{\mathfrak{sl}}

\newcommand{\zenter}{\mathcal{Z}}
\newcommand{\F}{\mathcal{F}}

\title{Stability in Bi-Hamiltonian Systems and Multidimensional Rigid Body}

\author{Anton Izosimov\footnote{Dept. of Mechanics and Mathematics, Moscow State University, e-mail: a.m.izosimov@gmail.com}}
\date{}
\begin{document}

\maketitle
%\begin{frontmatter}

%\ead{izosimov@mech.math.msu.su}
%\address{Chair of Differential Geometry and Applications, Dept. of Mechanics and Mathematics, Moscow State University, Russia, Moscow, 119991, Leninskie Gori, 1}

\begin{abstract}
The presence of two compatible Hamiltonian structures is known to be one of the main, and the most natural, mechanisms of integrability.
For every pair of Hamiltonian structures, there are associated conservation laws (first integrals). Another approach is to consider the second Hamiltonian structure on its own as a tensor conservation law. The latter is more intrinsic as compared to scalar conservation laws derived from it and, as a rule, it is ``simpler''. Thus it is natural to ask: can the dynamics of a bi-Hamiltonian system be understood by studying its Hamiltonian pair, without studying the associated first integrals?\par
In this paper, the problem of stability of equilibria in bi-Hamiltonian systems is considered and it is shown that the conditions for nonlinear stability can be expressed in algebraic terms of linearization of the underlying Poisson pencil. This is used to study stability of stationary rotations of a free multidimensional rigid body.
\end{abstract}
%
%\begin{keyword}
%bi-Hamiltonian systems \sep stability \sep Energy-Casimir method \sep multidimensional rigid body
%
%\MSC 37K10 \sep 37J25
%\end{keyword}
%
%\end{frontmatter}

\section{Introduction}
Since the pioneering works \cite{Magri, GD, Reiman}, the presence of two compatible Hamiltonian structures is known to be one of the main, and the most natural, mechanisms of integrability. This mechanism is responsible for the integrability of many equations coming from mechanics, mathematical physics and geometry (see, for example, \cite{AH} and references therein). The idea is that for every pair of Hamiltonian structures, there are associated conservation laws (first integrals). \par

Accordingly, a bi-Hamiltonian structure is usually considered as a ``factory of conservation laws''. However, the second Hamiltonian structure on its own can be considered as a tensor conservation law. The latter is more intrinsic as compared to scalar conservation laws derived from it and, as a rule, it is ``simpler''. For example, the second Poisson structure for the Korteweg-de Vries equation \cite{Magri} is linear, while the first integrals are complicated polynomials given by a recurrence formula. Thus it is natural to ask: can the dynamics of a bi-Hamiltonian system be understood by studying its Hamiltonian pair, without studying the associated first integrals?\par

In this paper, the problem of stability of equilibria in bi-Hamiltonian systems is considered and it is shown that the conditions for \textbf{nonlinear} stability in the bi-Hamiltonian case can be expressed in terms of the \textbf{linear} part of the underlying Poisson pencil. This linear part appears to be a collection of Lie algebras, each carrying a two-cocycle. Thus, the problem of stability in bi-Hamiltonian systems can be considered as algebraic. \par

As it was noted above, the notions ``bi-Hamiltonian''  and ``integrable'' are closely related. For this reason, the method discussed in this paper can be viewed as a rather general method for stability investigation in integrable systems. This paper focuses on the finite-dimensional case, however an essential part of the construction works in infinite dimension as well\footnote{Certainly, the infinite-dimensional case needs a separate discussion.}. Also note that the theorem formulated in this paper can be easily generalized from equilibria to periodic and general quasi-periodic trajectories.\par

To the author's knowledge, the idea of studying dynamics by means of a bi-Hamiltonian structure was first suggested by A.V.Bolsinov. In his paper \cite{Bolsinov} bi-Hamiltonian structure is used to describe the singular set of an integrable Hamiltonian system. 
Further developments are presented in the papers \cite{biham, SBS} devoted to a more detailed analysis of singularities of bi-Hamiltonian systems. In particular, in \cite{SBS} the notion of linearization of a Poisson pencil is introduced. This notion is used in the present paper to express stability conditions. \par

As an application, the stability problem for stationary rotations of a free multidimensional rigid body is solved. On the one hand, this problem is too complicated to be solved by a direct method (such as the Arnold method, see below), because the first integrals are polynomials of high degree. On the other hand, the bi-Hamiltonian structure of this problem is simple enough (in other words, the problem has complicated scalar conservation laws, but simple tensor conservation laws). This circumstance makes the application of the bi-Hamiltonian approach to a multidimensional rigid body extremely effective and provides a simple method for the determination of its stability.

\section{The Arnold method}
In the Hamiltonian case, stability in a linear approximation is always neutral and thus insufficient for a conclusion about nonlinear stability. To prove nonlinear stability in a Hamiltonian system, one usually uses the Arnold method (also known as the Energy-Casimir method, see \cite{Arnold}).  
The Arnold method can be formulated as follows: \begin{theorem}{Consider a Hamiltonian system $v$ on a Poisson manifold. Let $x$ be an equilibrium point of $v$ which belongs to a generic symplectic leaf $O$. Then $x$ is a critical point for the restriction of the energy to $O$. If this critical point is a non-degenerate minimum or maximum, then $x$ is stable.}\end{theorem}
Since we deal with the finite-dimensional case, this method proves the nonlinear (Lyapunov) stability.\par
Now let the system under consideration be integrable. Then one can replace the energy in the formulation of the Arnold method by any linear combination of the conserved quantities. 
The extended method can be formulated as follows:\begin{theorem}\label{EAEC}{
Consider an integrable Hamiltonian system $v$ on a Poisson manifold. Let $x$ be an equilibrium point of $v$ which belongs to a generic symplectic leaf $O$. Let also $f_1, \dots, f_n$ be the first integrals of the system. Suppose that there exists a linear combination $f = \sum a_if_i$ such that $\diff f|_O = 0$ and $\diff^2f|_O > 0$. Then $x$ is stable.}\end{theorem}
This extended formulation of the Arnold method is a powerful tool for investigating stability in integrable Hamiltonian systems. Nevertheless, in many dimensions it may be very complicated to compute the second differentials $\diff^2f_i$ and find a linear combination of them satisfying the conditions of the theorem. However, it turns out, that in the bi-Hamiltonian case this calculation can be replaced by a verification of a certain algebraic condition.

\section{Definitions}
\subsection{Poisson pencils and bi-Hamiltonian vector fields.}
\begin{definition}
	Two Poisson brackets on a manifold $M$ are called \textit{compatible}, if any linear combination of them is a Poisson bracket again.
The \textit{Poisson pencil} generated by two compatible Poisson brackets $P_0, P_\infty$ is the set
\begin{align}\label{pencil}\Pi = \{P_\lambda = P_0 - \lambda P_\infty\}_{\lambda \in \RP}.\end{align}
\end{definition}
Sometimes it also makes sense to consider $\lambda \in \CP$. 
\begin{remark}
	A Poisson pencil can also be defined as the set of all non-trivial linear combinations of two compatible brackets. However, it makes sense to consider Poisson brackets only up to proportionality, thus these two definitions may be considered as equivalent.\par
	The minus sign in (\ref{pencil}) is conventional.
\end{remark}

\begin{example}\label{argShiftMethod}
	Let $\g$ be an arbitrary Lie algebra, $a \in \g^{*}$. Consider the Lie-Poisson bracket\footnote{Throughout the whole paper Poisson brackets are identified with their Poisson tensors.} given by $P_0(x)(\xi, \eta) = x( [\xi, \eta] )$ and a constant bracket given by $P_\infty(x)(\xi, \eta) = a( [\xi, \eta] )$, the so-called bracket ``with the frozen argument''. It is easy to see that $P_0$ and $P_\infty$ are compatible. \par
	These two brackets are related to the so-called ``argument shift method'' introduced by A.S.Mishchenko and A.T.Fomenko \cite{MF}.
\end{example}	

\begin{definition} A vector field is bi-Hamiltonian with respect to a given pencil if it is Hamiltonian with respect to \textbf{all} brackets of the pencil. \end{definition}

\subsection{Rank and spectrum of a Poisson pencil.}

\begin{definition}
The \textit{rank of a pencil $\Pi$ at a point $x$} is the number
\begin{align}
\rank \Pi(x) = \max_{\lambda \in \CP} \rank P_\lambda(x).
\end{align}
The \textit{rank of a pencil $\Pi$} (on a manifold $M$) is the number
\begin{align}
\rank \Pi = \max_{x \in M} \rank \Pi(x).
\end{align}
\end{definition}

\begin{definition}
	The \textit{spectrum} of a pencil $\Pi$ at a point $x$ is the set
	\begin{align}
	\Lambda_\Pi(x) = \{ \lambda \in \CP \mid \rank P_\lambda(x) < \rank \Pi(x)\}.
	\end{align}
\end{definition}

\begin{example}\label{argShiftMethod2}
	Let $\Pi$ be the pencil from Example \ref{argShiftMethod}. If $a$ is regular, then the spectrum $\Lambda_\Pi(x)$ consists of $\lambda \in \Complex$ such that $x-\lambda a$ is singular in $\Complex \otimes \g^*$. If $a$ is singular, then the spectrum additionally contains $\lambda = \infty$.
\end{example}

\subsection{Linear Poisson pencils.}
\begin{definition}
	Let $\g$ be a Lie algebra and $A$ be a skew-symmetric bilinear form on it. 
	Then $A$ can be considered as a Poisson tensor on the dual space $\g^{*}$. Assume that the corresponding bracket is compatible with the Lie-Poisson bracket. In this case the Poisson pencil $\Pi(\g, A)$ generated by these two brackets is called the \textit{linear pencil} associated with the pair $(\g, A)$.
\end{definition}
	\begin{example}
	The pencil from Example \ref{argShiftMethod} is linear.
	\end{example}

The following is well known.
\begin{statement}
	A form $A$ on $\g$ is compatible with the Lie-Poisson bracket if and only if this form is a Lie algebra $2$-cocycle, i.e. 
	\begin{align}\label{cocycleCond}
		\diff A(\xi,\eta,\zeta) = A([\xi,\eta],\zeta) + A([\eta,\zeta],\xi) + A([\zeta,\xi],\eta) = 0
	\end{align}
	for any $\xi, \eta, \zeta \in \g$.
\end{statement}

\subsection{Linearization of a Poisson pencil.}
Let $P$ be a Poisson bracket. It is well-known that the linear part of $P$ at a point $x$ defines a natural Lie algebra structure on $\Ker P(x)$. This Lie algebra is called the linearization of $P$ at $x$. Now consider a Poisson pencil $\Pi = \{ P_{\lambda}\}$ and fix a point $x$. Denote by $\g_{\lambda}(x)$ the linearization of $P_\lambda$ at the point $x$.\par
It turns out that apart from the Lie algebra structure $\g_{\lambda}$ carries one more additional structure.
\begin{statement}
\begin{enumerate}
	\item For any $\alpha$ and $\beta$  the restrictions of $P_{\alpha}(x), P_{\beta}(x)$ on $\g_{\lambda}(x)$ coincide up to a multiplicative constant.
	\item The $2$-form $P_{\alpha}|_{\g_{\lambda}}$ is a $2$-cocycle on $\g_{\lambda}$.
\end{enumerate}
\end{statement}

Consequently, $P_{\alpha}|_{\g_{\lambda}}$ defines a linear Poisson pencil on $\g_{\lambda}^{*}$. Since $P_{\alpha}|_{\g_{\lambda}}$ is defined up to a multiplicative constant, the pencil is well-defined. Denote this pencil by $\diff_{\lambda} \Pi(x)$.
\begin{definition}
	The pencil $\diff_{\lambda} \Pi(x)$ is called the \textit{$\lambda$-linearization} of the pencil $\Pi$ at $x$.
\end{definition}
The linearization of a Poisson pencil at a given point is, therefore, not a single pencil, but a whole ``curve'' of linear Poisson pencils parametrized by $\lambda \in \CP$. However, if $\rank \Pi(x) = \rank \Pi$, then  it is easy to see that $\diff_{\lambda} \Pi(x)$ is non-trivial only for $\lambda \in \Lambda_\Pi(x)$.

\begin{example}
	Consider the pencil from Example \ref{argShiftMethod}. The algebra $\g_\lambda(x)$ in this case is simply the stabilizer of $x-\lambda a$. The second form $P_{\alpha}|_{\g_{\lambda}}$ is given on the stabilizer by the formula $a|_{g_\lambda}( [\xi, \eta] )$. Thus, the $\lambda$-linearization is  the ``restriction'' of the initial pencil to the stabilizer of $x-\lambda a$. If $\lambda$ is not in the spectrum, then this stabilizer is abelian, and the linearization is trivial. 
\end{example}
\begin{remark}
Note that it is natural to expect that a ``linearization'' of an object defined on a manifold $M$ is an object defined on the tangent space $\T_xM$. For a $\lambda$-linearization of a Poisson pencil this is not so: it is defined on $(\Ker P_\lambda(x))^*$. However, the natural inclusion map $\Ker P_\lambda(x) \to \T^*_xM$ induces an isomorphism
\begin{align}
	\T_xM / \T_xO_\lambda(x) \simeq (\Ker P_\lambda(x))^*,
\end{align}
where $O_\lambda(x)$ is the symplectic leaf of $P_\lambda$ passing through $x$. Thus, $\diff_{\lambda} \Pi(x)$  can be considered as a Poisson pencil on the quotient $\T_xM / \T_xO_\lambda(x)$.
\end{remark}

\subsection{Compact linear pencils.}
Let $A$ be a $2$-cocycle on a Lie algebra $\g$.
For an arbitrary element $\nu \in \Ker A$ define the bilinear form $A^\nu(\xi, \eta) = A([\nu, \xi], \eta)$. The cocycle identity (\ref{cocycleCond}) implies that this form is symmetric. Furthermore, $\Ker A^\nu \supset \Ker A$, therefore $A^\nu$ is a well-defined symmetric form on $\g / \Ker A$.

 \begin{definition}
 A linear pencil $\Pi(\g, A)$ is \textit{compact} if there exists $\nu \in \zenter(\Ker A)$ such that $A^{\nu}$ is positive-definite on $\g / \Ker A$. 

\end{definition}
\begin{remark}
$\zenter$ stands for the center of a Lie algebra.
\end{remark}

\begin{example}
	Any linear pencil on a compact semisimple Lie algebra is compact.
	 Indeed, let $\g$ be a compact semisimple Lie algebra. Since $\Hom^2(\g) = 0$, any cocycle $A$ on $\g$ has the form $A(\xi, \eta) = \langle a,[\xi, \eta] \rangle$, where $\langle\, , \rangle$ is the Killing form. It is easy to see that for $\nu = a$ the form $A^\nu$ is positive-definite on $\g/\Ker A$.
\end{example}
\begin{example}
	Let $\g =  \sL(2,\R)$. Again, any cocycle on $\g$ has the form  $A(\xi, \eta) = \langle a,[\xi, \eta] \rangle$. Suppose that $a \neq 0$. Then it is easy to see that $\Pi(\g,A)$ is compact if and only if the Killing form is negative on $a$. A suitable choice of $\nu$ is $\nu = -a$.
\end{example}

\begin{example}
	Let $\g = \mathrm{Vect}(\mathrm{S}^1)$ be a Lie algebra of vector fields on a circle and $A$ be the Gelfand-Fuks cocycle (see \cite{AH}):
	\begin{align}
	A(\phi, \psi) = \int\limits_0^{2\pi} \phi\psi''' \diff x.
	\end{align}
	Then the pencil $\Pi(\g, A)$ is compact. Indeed, if we choose $\nu = 1$, then 
	\begin{align}
		A^\nu(\phi, \phi) =  \int\limits_0^{2\pi} (\phi'')^2 \diff x.
	\end{align}
\end{example}

\subsection{Geometric meaning of compactness condition.}
\begin{statement}
	Suppose that a system $v$ is bi-Hamiltonian with respect to a compact linear pencil. Then the origin is a stable equilibrium of $v$.
\end{statement}

The idea of the proof is that the form $A^\nu > 0$ can be used to construct a positive-definite integral of $v$.\par
A similar statement is true for nonlinear pencils, see Theorem \ref{stabThm}. In this case stability can be studied by checking the compactness of the \textbf{linearizations}. In contrast to the classical linearization procedure, which can only prove linearized stability, bi-Hamiltonian linearization (defined above) proves \textbf{nonlinear} stability.

\subsection{Diagonalizability condition.}

\begin{definition}\label{RegEq}
	The pencil $\Pi$ is called diagonalizable at $x$ if
	\begin{align}
	\dim \Ker \left( P_\alpha(x)|_{P_\lambda(x)}\right) = \corank \Pi(x) \mbox{ for all } \lambda \in \Lambda_\Pi(x), \alpha \neq \lambda.
	\end{align}
\end{definition}
\begin{remark}\label{JK}
The Jordan-Kronecker theorem (see \cite{GZ2}) claims that two skew-symmetric forms on a vector space can be simultaneously brought to a certain block-diagonal form. This form contains blocks of two types: Jordan blocks and Kronecker blocks. 
The diagonalizability condition means that all Jordan blocks for $P_0(x), P_\infty(x)$ have size $2 \times 2$. 
\end{remark}
	\begin{example}
	Let $\Pi$ be the pencil from Example \ref{argShiftMethod}. Suppose that $a$ is regular. Then the pencil is diagonalizable at $x$ if for each $\lambda \in \Lambda_\Pi(x)$ the following two conditions hold:
	\begin{enumerate}
		\item The index\footnote{Recall that the index of a Lie algebra $\g$ can be defined as the corank of the corresponding Lie-Poisson structure, or, equivalently, as the dimension of the stabilizer of a regular element $a \in \g^*$.} of the stabilizer of $x - \lambda a$ equals the index of $\g$.
		\item The restriction of $a$ to the stabilizer of $x - \lambda a$ is a regular element.
	\end{enumerate}
	\end{example}

\subsection{Regularity condition.}
Let $v$ be a system which is bi-Hamiltonian with respect to $\Pi$, $v(x) = 0$.
Suppose that $\rank \Pi(x) = \rank \Pi$.
\begin{definition}\label{regEq}
Say that $x$ is \textit{regular} if the following condition holds:
	\begin{align}
	\Ker P_\alpha(x) = \Ker P_\beta(x) \mbox{ for all } \alpha, \beta \notin \Lambda(x).
	\end{align}
\end{definition}

\begin{theorem}[A.V.Bolsinov, A.A.Oshemkov \cite{biham}]
	Let $v$ be a system which is bi-Hamiltonian with respect to $\Pi$, $v(x) = 0$. Suppose that $\rank \Pi(x) = \rank \Pi$. Then, if $x$ is not regular, we can find an integral $f$ of $v$ and $\alpha \in \R$ such that $P_\alpha \diff f(x) \neq 0$.
\end{theorem}
Consequently, if $x$ is not regular, the whole trajectory of $P_\alpha\diff f$ passing through $x$ consists of equilibrium points of $v$. In this situation it can be shown that, provided the system is non-resonant\footnote{Recall that an integrable system is called non-resonant if its trajectories are dense on almost all Liouville tori. See \cite{intsys}.}, $x$ cannot be Lyapunov stable. Therefore, it only makes sense to study regular equilibria for stability.

\section{Stability theorem}

\begin{theorem}[Stability theorem]
\label{stabThm}
Suppose that $\Pi$ is a Poisson pencil on a finite-dimensional manifold, $v$ is bi-Hamiltonian with respect to $\Pi$. Let $x$ be an equilibrium of $v$. Assume that  
 \begin{enumerate}
	\item $\rank \Pi(x) = \rank \Pi$.
	\item The equilibrium $x$ is regular.
%	\item The spectrum of $\Pi$ at $x$ is real\,\footnote{Formally, this condition follows from condition 5, because a linear pencil on a complex Lie algebra cannot be compact.}: $\Lambda_\Pi(x) \subset \RP$.
	\item The pencil $\Pi$ is diagonalizable at $x$.
	\item For each $\lambda \in \Lambda_\Pi(x)$ the $\lambda$-linearization $\diff_\lambda \Pi(x)$ is compact.
 \end{enumerate}
 Then $x$ is Lyapunov (nonlinearly) stable.
 \end{theorem}
The proof is given in Section \ref{proof}.
\begin{remark}\label{realSpectrum}
Condition 4 implies that the spectrum of $\Pi$ at $x$ is real, since a pencil on a complex Lie algebra cannot be compact.
\end{remark}
Theorem \ref{stabThm} is a bi-Hamiltonian reformulation of Theorem \ref{EAEC} in the following sense. Let a system $v$ be bi-Hamiltonian with respect to a pencil $\Pi$. Then the Casimir functions of all brackets of the pencil are first integrals of $v$. These first integrals are known to be in involution (see \cite{Reiman}). Consider the family $\F$ generated by all these first integrals. Then the following is true:
\textit{if the  first condition of Theorem \ref{stabThm} is satisfied, then the subsequent conditions are equivalent to the existence of  $f \in \F$ satisfying the conditions of Theorem \ref{EAEC}}.\par
If $\F$ happens to exhaust all the first integrals of $v$, then Theorems \ref{EAEC} and \ref{stabThm} are equivalent (for generic points satisfying $\rank \Pi(x) = \rank \Pi$). This should be expected if the first integrals belonging to $\F$ are sufficient for complete Liouville integrability of $v$.
	\begin{definition}\label{KroneckerDef}
		A pencil is called \textit{Kronecker} if its spectrum is empty almost everywhere.
	\end{definition}
	\begin{remark}
		This condition means that the Jordan-Kronecker normal form (see Remark \ref{JK}) for $P_0(x), P_\infty(x)$ contains only Kronecker blocks for almost all $x$.
	\end{remark}
%	Assume that $v$ is bi-Hamiltonian with respect to a pencil $\Pi$. Then the Casimir functions of all brackets of the pencils are integrals of $v$. These integrals are known to be in involution (see \cite{Reiman}).
	\begin{theorem}[A.V.Bolsinov \cite{Bolsinov}]
	Let $v$ be bi-Hamiltonian with respect to a pencil $\Pi$ and $\F$ be the family of first integrals of $v$ described above. Then the first integrals belonging to $\F$ are sufficient for complete Liouville integrability of $v$ if and only if $\Pi$ is Kronecker.
	\end{theorem}	
	%Suppose that this condition is satisfied. If, additionally, $v$ is non-resonant, then any integral $f$ of $v$ is constant on all Liouville tori and, consequently, is functionally dependent with the functions belonging to $\F$. Thus, it is natural to expect that all integrals 
	So, Theorem \ref{stabThm} should be the most effective for Kronecker pencils.

 \section{Multidimensional rigid body}
 \subsection{Statement of the problem.}
 It is well known that a free asymmetric three-dimensional rigid body admits three stationary rotations\footnote{A rotation is called stationary if the axis of rotation is time independent.}. These are the rotations around three principal axes of inertia. The rotations around the long and the short axes are stable, while the rotation around the intermediate axis is unstable (see \cite{Arnold}). The problem is to obtain a multidimensional generalization of this fact, i.e. to study stationary rotations of a free multidimensional rigid body for stability.\par
 This problem has been studied by many people, see \cite{Oshemkov, Marshall, Spiegler, Ratiu, Casu2, Casu}. However, no general solution is known.

 \subsection{The Euler-Arnold equations and the bi-Hamiltonian structure.}
 The dynamics of the angular velocity matrix $\Omega$ of a free multidimensional rigid body is described by the Euler-Arnold equations (see \cite{Arnold})
  \begin{align}\label{MRB}
\dot \Omega J + J \dot \Omega = [J, \Omega^2], 
 \end{align}
 where $J$ is the mass tensor (see below).
 \begin{remark}
 The proof of integrability of (\ref{MRB}) belongs to S.V.Manakov\,\cite{Manakov}. For this reason the system (\ref{MRB}) is also known as the Manakov top.
 \end{remark}
 The following two observations allow the application of Theorem \ref{stabThm} to the problem of stability of stationary rotations:
 \begin{enumerate}
 \item Stationary rotations are just the equilibria of (\ref{MRB}).
 \item The system (\ref{MRB}) is bi-Hamiltonian (with respect to a Kronecker pencil), as it was observed by A.V.Bolsinov\,\cite{Bolsinov}.
The first Poisson structure (due to Arnold) is the standard Lie-Poisson structure on $\so(n)^*$. The second (due to Bolsinov) is also a Lie-Poisson structure, but for a non-standard commutator on $\so(n)$ given by $[X,Y] = XJ^2Y - YJ^2X$.
\end{enumerate}
 \subsection{Rotation of a multidimensional body.}
 First, consider how an $n$-dimensional body may rotate. At each moment of time $\R^n$ is decomposed into a sum of $m$ pairwise orthogonal two-dimensional planes $\Pi_1, \dots, \Pi_m$ and a space $\Pi_0$ of dimension $n-2m$ orthogonal to all these planes:
\begin{align}\label{decomp}
\R^n = \left(\bigoplus_{i=1}^{m} \Pi_i \right)\oplus \Pi_0.
\end{align}
 There is an independent rotation in each of the planes $\Pi_1, \dots, \Pi_m$, while $\Pi_0$ is fixed\footnote{Note that $\Pi_0$ may be zero in the even-dimensional case, which means that there are no fixed axes.}. In other words, a rotation of a multidimensional body can be represented as a superposition of ``elementary'' $2$-dimensional rotations. \par
A rotation is stationary if all the planes $\Pi_0, \dots, \Pi_m$ are time independent (this condition automatically implies that the velocities of the rotations are also constant). \par
Before studying stationary rotations for stability it is necessary to find these rotations. Recall that a rotation of a generic three-dimensional rigid body is stationary if and only if it is a rotation around one of the principal axes of inertia. In the multidimensional case the situation is slightly more complicated. If the planes $\Pi_0, \dots, \Pi_m$ are spanned by principal axes of inertia (such rotations are called in \cite{NRE} \textit{regular}), then the rotation is stationary. But the converse is not necessarily true (see \cite{NRE}). However, as it is shown in \cite{MBR}, the rotations which are not regular are always unstable. Therefore, it is only necessary to consider regular stationary rotations.
\subsection{Mass tensor of a rigid body.}
From the dynamical point of view a rigid body is characterized by its mass tensor $J$. The entries of this tensor are given by
\begin{align}
J_{ij} = \int (x_i - \widehat x_i)(x_j - \widehat x_j) \diff \mu,
\end{align}
where $\widehat x_i$ are the center of mass coordinates.\par
A body is called asymmetric if all the eigenvalues of $J$ are distinct.
\subsection{Parabolic diagram of a regular stationary rotation.}
Consider a regular stationary rotation. Then, by definition, the planes $\Pi_i$ entering (\ref{decomp}) are spanned by principal axes of inertia. For each plane $\Pi_i, i > 0$ let us denote by $\lambda_1(\Pi_i), \lambda_2(\Pi_i)$  the eigenvalues of the mass tensor $J$ corresponding to the principal axes of inertia which span $\Pi_i$. By $\omega(\Pi_i)$ denote the angular velocity of rotation in the plane $\Pi_i$.\par
Draw a coordinate plane. Mark squares of all eigenvalues of $J$ on the horizontal axis. For each $\Pi_i$ draw a parabola through $ \lambda_1(\Pi_i)^{2},  \lambda_2(\Pi_i)^{2}$ given by $y = \chi_{i}(x)$, where
		\begin{align}\label{chiFunc}
			\chi_{i}(x) = \frac{(x - \lambda_1(\Pi_i)^{2})(x - \lambda_{2}(\Pi_i)^{2})}{\omega(\Pi_{i})^{2}( \lambda_{1}(\Pi_i) +  \lambda_{2}(\Pi_i))^{2}}.
		\end{align}
		For all fixed principal axes draw vertical lines through the squares of corresponding eigenvalues of $J$.
\begin{definition}
		The obtained picture is called \textit{the parabolic diagram} of a regular stationary rotation.
	\end{definition}
	Figures \ref{pd1}, \ref{pd2} illustrate two examples of parabolic diagrams. \par
\begin{figure}[t]
\centerline{ \includegraphics[scale=0.4]{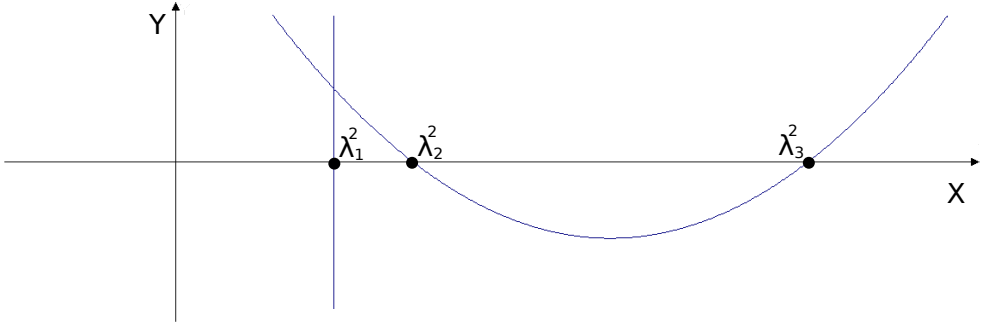}}
%%%\quad\includegraphics[scale=1]{6d.eps}}
\caption{Stable rotation of a three-dimensional body around the short principal axis of inertia. $\lambda_1 < \lambda_2 < \lambda_3$ are the eigenvalues of the mass tensor.}\label{pd1}
\end{figure}

\smallskip
\begin{definition}
\begin{enumerate}
	\item Two parabolas on a parabolic diagram are said to intersect at infinity if they have only one point of intersection (of multiplicity one) or no points of intersection (neither real, nor complex).
	\item Two parabolas on a parabolic diagram are said to be tangent at infinity if they have no points of intersection (real or complex). 
\end{enumerate}
\end{definition}
% \begin{definition}
%	Say that a parabolic diagram is \textit{generic} if all intersections on it are simple, i.e. it contains no multiple intersections and no points of tangency.
%\end{definition}
\begin{figure}[b]
\centerline{\includegraphics[scale=0.4]{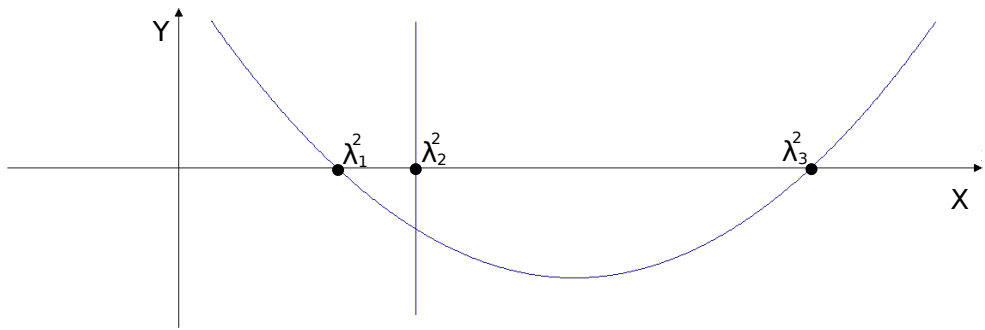}}
\caption{Unstable rotation of a three-dimensional body around the middle principal axis of inertia.}\label{pd2}
\end{figure}
\subsection{Stability theorems.}

Applying Theorem \ref{stabThm} we obtain the following result:

  \begin{theorem}\label{newStabThm}
 Consider a regular stationary rotation of an asymmetric multidimensional rigid body.  Assume that
 \begin{enumerate}
  \item All intersections on the parabolic diagram of the rotation are either real and belong to the upper half-plane or infinite.
  \item There are no points of tangency on the parabolic diagram.
     \item The rotation has no more than two fixed axes ($\dim \Pi_0 \leq 2$).
  \end{enumerate}
Then the rotation is stable.
 \end{theorem}
 \begin{remark}
 Vice versa, if the parabolic diagram of a rotation contains at least one complex intersection point or an intersection at the lower half-plane, then the rotation is unstable. This is proved in \cite{MBR}.
 \end{remark}
 For example, the rotation illustrated in Figure \ref{pd1} is stable, in Figure \ref{pd2} -- unstable.\par
 \smallskip
 The proof of Theorem \ref{newStabThm} is the formal application of  Theorem \ref{stabThm}. Below are some brief comments on how the conditions of Theorem \ref{newStabThm} are related to the conditions of Theorem \ref{stabThm}.
 \begin{enumerate}
 \item Condition 1 of Theorem\,\ref{stabThm} which reads ``$\rank \Pi(x) = \rank \Pi$'' is equivalent to the condition that ``the rotation has no more than two fixed axes''.
 \item Condition 2 of Theorem\,\ref{stabThm} which reads ``the equilibrium is regular'' is equivalent to the fact that a rotation is regular.
 \item The spectrum of the pencil is exactly the set of the horizontal coordinates of the intersection points on the parabolic diagram. Thus, parabolic diagrams naturally appear in the problem.
  \item Condition 3 of Theorem\,\ref{stabThm} which reads  ``the pencil is diagonalizable'' is equivalent to the condition that ``there are no points of tangency on the parabolic diagram''.
  \item Condition 4 of Theorem\,\ref{stabThm} which reads ``for each $\lambda \in \Lambda_\Pi(x)$ the $\lambda$-linearization $\diff_\lambda \Pi(x)$ is compact'' is equivalent to the condition that ``All intersections on the parabolic diagram of the rotation are either real and belong to the upper half-plane or infinite''. 
 \end{enumerate}
 
Note that parabolic diagrams, which appear naturally as spectral data of the Poisson pencil associated with a rigid body, give a visual interpretation of stability results even in the four-dimensional case, which was studied earlier by direct methods in \cite{Marshall, Ratiu, Casu2}.\par
Thus the bi-Hamiltonian approach, in this case, not only allows simpler calculations but also provides a more natural interpretation of the results. See also \cite{AS} where the stability problem for the multidimensional rigid body is solved by means of algebraic geometry.

\section{Proof of the stability theorem}\label{proof}
A technique similar to that used in \cite{SBS} will be used to prove Theorem \ref{stabThm}, however the proof is self-contained.

\subsection{Step 1. The forms $Q_f$.}
For notational simplicity denote the spectrum of $\Pi$ at $x$ by $\Lambda$ and the cotangent space to the ambient manifold $M$ by $V^*$.\par
Suppose that the equilibrium point $x$ is regular. Then, by definition, $\Ker P_\alpha = \Ker P_\beta$ for all $\alpha, \beta \notin \Lambda$. Denote this common kernel by $K$. The regularity condition also implies that the symplectic leafs of all brackets $P_\alpha, \alpha \notin \Lambda$ have a common tangent space at the point $x$. Denote this tangent space by $O$. It will be proved that under the conditions of Theorem \ref{stabThm} there exists an integral $f$ such that $\diff^2f|_O(x) > 0$. If this is so, then stability follows from Theorem \ref{EAEC}.\par
Without loss of generality assume that $\infty \notin \Lambda$. Then the map $P_\infty \colon V^*/K \to O$ is an isomorphism. Instead of $\diff^2 f$ consider the form $Q_f$ defined on $V^* / K$ by
\begin{align}
	Q_f(\xi,\eta) = \diff^2 f(P_\infty \xi, P_\infty \eta).
\end{align}
Obviously, $\diff^2 f$ and $Q_f$ are both simultaneously positive definite.\par
Let $\alpha_0$ be larger than all elements of $\Lambda$. Define $\F$ as the space spanned by all (local) Casimir functions of all brackets $P_\alpha$, where $\alpha_0 < \alpha < +\infty$. It will be proved that there exists $f \in \F$ such that $Q_f > 0$ on $V^*/K$.\par
\subsection{Step 2. Decomposition of $V^*/K$.} Since $\Ker P_\alpha \supset K$ for each $\alpha$, all forms $P_\alpha$ are well defined on $V^*/K$. Since $P_\infty$ is non-degenerate on $V^*/K$, consider the recursion operator
\begin{align}
R = P_\infty^{-1}P_0 \colon V^*/K \to V^*/K.
\end{align}
It is easy to see that the spectrum of $R$ coincides with the spectrum of the pencil: $\sigma(R) = \Lambda$. The $\lambda$-eigenspace of $R$ is \begin{align}V_\lambda = \Ker P_\lambda / K.\end{align}
\begin{lemma}\label{lemma1}
	Under the conditions of Theorem \ref{stabThm} the operator $R$ is diagonalizable over $\R$.  
\end{lemma}
\begin{proof}
	First, $\Lambda \subset \R$ (see Remark \ref{realSpectrum}). Consequently, all eigenvalues of $R$ are real and it suffices to prove that $R$ is diagonalizable. Suppose the contrary, i.e. that $R$ has a Jordan block. Then there exists $\xi \in V_\lambda, \eta \in V^*/K$ such that $R\eta = \lambda \eta + \xi$.
	Therefore, for all $\zeta \in V^*/K$, $P_\lambda(\eta,\zeta) = P_\infty(\xi,\zeta)$. Consequently, $P_\infty(\xi,\zeta) = 0$  for all $\zeta \in V_\lambda$. But the diagonalizability condition implies that $P_\infty$ is non-degenerate on $V_\lambda$. The obtained contradiction proves the lemma.
\end{proof}
	So, under the conditions of Theorem  \ref{stabThm} there exists a decomposition
	\begin{align}\label{decomp2}
		V^* / K = \bigoplus_{\lambda \in \Lambda} V_\lambda.
	\end{align}
	It will be shown that all forms $Q_f, f \in \F$ respect this decomposition.\par
\begin{lemma}\label{lemma2}
	For each $f \in \F$ and each $\alpha < \alpha_0$ there exist functions $\wave f, \bar f \in \F$ such that for any function $g$ the following ``recursion'' relations hold:
	\begin{align}
			\{f, g\}_\infty = \{\wave f,g\}_\alpha, \quad \{f, g\}_\alpha = \{\bar f,g\}_\infty. 
	\end{align}
\end{lemma}
\begin{proof}
	First, let $f$ be a Casimir function of $P_\beta, \beta > \alpha_0$. Then
	\begin{align}
	\{f, g\}_\alpha = \{f,g\}_\beta + (\beta - \alpha) \{f,g\}_\infty = (\alpha - \beta) \{f,g\}_\infty.
	\end{align}	
	Thus, $\wave f =f / (\beta - \alpha), \bar f = (\beta - \alpha)f$ are as required. \par
	For an arbitrary $f \in \F$ the statement is true by linearity.
\end{proof}
Let $D_fP_\alpha$ be the operator dual to the linearization of $P_\alpha\diff f$ at $x$. Then it is easy to see that $Q_f$ is given by the formula
\begin{align}\label{QFFormula}
	Q_f(\xi,\eta) = P_\infty(D_fP_\infty (\xi),\eta).
\end{align}
The operator $D_fP_\alpha \colon V^* \to V^*$ can be given by an explicit formula
\begin{align}\label{DFPFormula}
	D_fP_\alpha(\diff g(x)) =\diff \{ f,g\}_\alpha(x).
\end{align}
Note that this formula, together with the Jacobi identity, implies that $D_fP_\alpha$ is skew-symmetric with respect to $P_\alpha$
%Since $\diff f \in \Ker P_\alpha$, the value on the right hand-side does not depend on $g$, but only on $\diff g(x)$.
\begin{lemma}\label{lemma3}
	For $f \in \F$ the operator $D_fP_\infty$ is skew-symmetric with respect to all forms $P_\alpha$.
\end{lemma}
\begin{proof}
	First, $D_fP_\infty$ is skew-symmetric with respect to $P_\infty$. Taking into account (\ref{DFPFormula}) and Lemma \ref{lemma2}, $D_fP_\infty$ can be rewritten as $D_{\wave f}P_0$. Thus, this operator is skew-symmetric with respect to $P_\infty$ and $P_0$, and, by linearity, with respect to all forms of the pencil.
\end{proof}

\begin{lemma}\label{lemma4}
	For $f \in \F$ the recursion operator is symmetric with respect to $Q_f$: 
	\begin{align}
	Q_f(R\xi, \eta) = Q_f(\xi, R\eta),
	\end{align}
	and, consequently, the summands of (\ref{decomp2}) are pairwise orthogonal with respect to $Q_f$.
\end{lemma}
\begin{proof}
	Lemma \ref{lemma3} implies that $D_fP_\infty$ commutes with $R$. Also note that $R$ is symmetric with respect to $P_\infty$, and $D_fP_\infty$ is skew-symmetric with respect to $P_\infty$. Therefore
	\begin{align}
	\begin{aligned}
		Q_f(R\xi,\eta) &= P_\infty(D_fP_\infty (R\xi),\eta) = -P_\infty(R\xi,D_fP_\infty \eta) = \\ &= -P_\infty(\xi ,RD_fP_\infty(\eta)) = P_\infty (D_fP_\infty (R\eta),\xi) = Q_f(R\eta, \xi).
	\end{aligned}
	\end{align}
\end{proof}
  
 \subsection{Step 3. Positivity of $Q_f$ on $V_\lambda$.}

\begin{lemma}\label{lemma5}
	Under the conditions of Theorem \ref{stabThm} for each $\lambda \in \Lambda$ there exists $f_\lambda \in \F$ such that $Q_{f_\lambda}$ is positive on $V_\lambda$.
\end{lemma}
\begin{proof}
Let $[\,,]$ be the commutator in $\g_\lambda = \Ker P_\lambda$. The compactness condition implies that there exists $\nu \in K$ such that 
	\begin{align}
	A^{\nu}(\xi, \xi) = P_\infty([\nu,\xi], \xi)
	\end{align}
	is positive definite on $V_\lambda$. Take $f \in \F$ such that $\diff f(x) = \nu$. By Lemma \ref{lemma2} there exists $\bar f \in \F$ such that $\{f, g\}_\lambda = \{\bar f,g\}_\infty$. Take $f_\lambda = \bar f$. Let $\xi \in V_\lambda$. Then
	 \begin{align}
	 	Q_{\bar f}(\xi,\xi) = P_\infty(D_{\bar f}P_\infty (\xi),\xi) = P_\infty(D_{f}P_\lambda (\xi),\xi).
	 \end{align}
	 Since $\xi \in \Ker P_\lambda$, (\ref{DFPFormula}) implies that $D_{f}P_\lambda (\xi) = [\diff f(x), \xi]$. Thus,
	 \begin{align}
	 	Q_{\bar f}(\xi,\xi) = P_\infty([\diff f, \xi],\xi) = P_\infty([\nu, \xi],\xi) = A^{\nu}(\xi, \xi) > 0,
	 \end{align}	 
\end{proof}

\subsection{Step 4. Recursion invariance.}
\begin{lemma}\label{lemma6}
	Suppose that $f \in \F$ and $p(z)$ is a polynomial. Then there exists $\wave f \in \F$ such that
	\begin{align}
	Q_f(p(R)\xi, \eta) = Q_{\wave f}(\xi,\eta).
	\end{align}	
\end{lemma}
\begin{proof}
First suppose that $p(z) = z$. Since $\diff f \in K$, there exists a Casimir function $f_\infty$ of $P_\infty$ such that $\diff f(x) = \diff f_\infty(x)$.
	Formula (\ref{DFPFormula}) implies that $D_fP_\infty = D_{f - f_\infty}P_\infty$.
	Further, for any $g$ such that $\diff g = 0$ the same formula (\ref{DFPFormula}) implies that $D_gP_\alpha = \diff^2g P_\alpha$, thus
	\begin{align}
		D_fP_\infty = \diff^2(f - f_\infty)P_\infty.
	\end{align}	
	Consequently, taking into account (\ref{QFFormula}),
	\begin{align}\label{lemma6f1}
	\begin{aligned}
	Q_f(R\xi,\eta) &= P_\infty( \diff^2(f - f_\infty)P_\infty R(\xi),\eta) = P_\infty( \diff^2(f - f_\infty)P_0(\xi),\eta) = \\ &= P_\infty( D_{f - f_\infty}P_0(\xi),\eta).
	\end{aligned}
	\end{align}
	Note that if $f - f_\infty \in \F$, then, by Lemma \ref{lemma2}, $D_{f - f_\infty}P_0(\xi)$ can be rewritten as $D_{\bar f}P_\infty$, which proves the lemma. However, $f - f_\infty$ is not in $\F$ a priori\footnote{This can be overcome by adding Casimir functions of $P_\infty$  to $\F$. However, this would make the proof of Lemma \ref{lemma2} much more complicated.}, therefore the following limit argument is applied.\par
	Choose a family $f_\alpha$ such that $f_\alpha$ is a Casimir function of $P_\alpha$ and $f_\alpha \to f_\infty$ as $\alpha \to \infty$.  Then $f-f_\alpha \in \F$, and, by Lemma \ref{lemma2}, there exists $\bar f_\alpha \in \F$ such that $\{f - f_\alpha, g\}_0 = \{\bar f_\alpha,g\}_\infty$. Thus, $D_{f - f_\alpha}P_0 = D_{\bar f_\alpha}P_\infty$.
	So (\ref{lemma6f1}) gives
	\begin{align}
		Q_f(R\xi,\eta) = \lim_{\alpha \to \infty} P_\infty( D_{f - f_\alpha}P_0(\xi),\eta) = \lim_{\alpha \to \infty}P_\infty( D_{\bar f_\alpha}P_\infty(\xi),\eta) = \lim_{\alpha \to \infty} Q_{\bar f_\alpha}(\xi, \eta).
	\end{align}
	Consequently, the form $Q_f(R\xi,\eta)$ belongs to the closure of the space $\{Q_g, g \in \F\}$. But this latter space is finite-dimensional, thus $Q_f(R\xi,\eta) = Q_{\wave f}(\xi,\eta)$ for some $\wave f \in \F$.\par
	For an arbitrary polynomial the lemma is proved by induction.

\end{proof}

\subsection{Step 5. Completion of the proof.}
Let 
\begin{align}
	p_{\lambda_0}(z) = \prod\limits_{\lambda \in \Lambda\setminus \{\lambda_0\}}\frac{z - \lambda}{\lambda_0 - \lambda}.
\end{align}
Then $p_\lambda(R)$ is the projector $V^*/K \to V_\lambda$. \par
By Lemma \ref{lemma5} there exists $f_\lambda \in \F$ such that $Q_{f_\lambda}$ is positive on $V_\lambda$. By Lemma \ref{lemma6} there exists $\wave f_\lambda \in \F$ such that
\begin{align}
	Q_{f_\lambda}(p_\lambda(R)\xi, \xi) = Q_{\wave f_\lambda}(\xi, \xi).
\end{align}
Take 
\begin{align}
	f = \sum\limits_{\lambda \in \Lambda} \wave f_\lambda.
\end{align}
Now claim that $Q_f > 0$ on $V^*/K$. By Lemma \ref{lemma4} it suffices to show that $Q_f$ is positive on each $V_\lambda, \lambda \in \Lambda$. Let $\xi \in V_{\lambda_0}$. Then
\begin{align}
	Q_f(\xi , \xi) = \sum\limits_{\lambda \in \Lambda} Q_{\wave f_\lambda}(\xi, \xi) = \sum\limits_{\lambda \in \Lambda} Q_{f_\lambda}(p_\lambda(R)\xi, \xi) = Q_{f_{\lambda_0}}(\xi, \xi) > 0.
\end{align}
In the last equality we used the fact that $p_{\lambda}(R)$ is a projector. The theorem is proved.
\section{Acknowledgements}
I am grateful to Alexey Bolsinov for fruitful discussions on the subject. I would also like to thank David Dowell for his useful comments.\par
%This work was supported by Loughborough University, Hausdorff Research Institute for Mathematics (Bonn), RFBR grant 10–01–00748-a, the Programme for the Support of Leading Scientific Schools of RF (grant NSh-1410.2012.1), the programme ``Scientific and Scientific-Pedagogical Personnel of Innovative Russia'' (grant 14.740.11.0794), and by the Government grant of the Russian Federation for support of research projects implemented by leading scientists, in the Federal State Budget Educational Institution of Higher Professional Education Lomonosov Moscow State University under the agreement No. 11.G34.31.0054. 
%\section*{}
\bibliographystyle{unsrt}  
\bibliography{StabBihamRevised2} 
\end{document}